\documentclass[a4paper,UKenglish,cleveref, autoref]{oasics-v2019}
\title{Selfish Mining and Dyck Words in Bitcoin and Ethereum Networks} %TODO Please add

\titlerunning{Selfish Mining and Dyck Words in Bitcoin and Ethereum Networks}%optional, please use if title is longer than one line

\author{Cyril Grunspan}{L{\'e}onard de Vinci, Research Center, Paris - La D{\'e}fense, France}{cyril.grunspan@devinci.fr}{}
{}%TODO mandatory, please use full name; only 1 author per \author macro; first two parameters are mandatory, other parameters can be empty. Please provide at least the name of the affiliation and the country. The full address is optional

\author{Ricardo P\'erez-Marco}{CNRS, IMJ-PRG, Univ. Paris 7, Paris, France}{ricardo.perez-marco@imj-prg.fr}{}
{}

\authorrunning{C. Grunspan and R. P\'erez-Marco}%TODO mandatory. First: Use abbreviated first/middle names. Second (only in severe cases): Use first author plus 'et al.'

\Copyright{Cyril Grunspan and Ricardo P\'erez-Marco}%TODO mandatory, please use full first names. LIPIcs license is "CC-BY";  http://creativecommons.org/licenses/by/3.0/

\ccsdesc[100]{Mathematics of Computing}%TODO mandatory: Please choose ACM 2012 classifications from https://dl.acm.org/ccs/ccs_flat.cfm 

\keywords{Bitcoin, Blockchain, Ethereum, Proof-of-Work,
Selfish Mining, Stubborn Mining, Apparent Hashrate, Revenue Ratio, Catalan
Distributions, Dyck Words, Random Walk.}%TODO mandatory; please add comma-separated list of keywords

\relatedversion{}%optional, e.g. full version hosted on arXiv, HAL, or other respository/website
\supplement{}%optional, e.g. related research data, source code, ... hosted on a repository like zenodo, figshare, GitHub, ...

\nolinenumbers %uncomment to disable line numbering

\EventEditors{Vincent Danos, Maurice Herlihy, Maria Potop-Butucaru, Julien Prat, and Sara Tucci-Piergiovanni}
\EventNoEds{5}
\EventLongTitle{International Conference on Blockchain Economics, Security and Protocols (Tokenomics 2019)}
\EventShortTitle{Tokenomics 2019}
\EventAcronym{Tokenomics}
\EventYear{2019}
\EventDate{May 6--7, 2019}
\EventLocation{Paris, France}
\EventLogo{}
\SeriesVolume{71}
\ArticleNo{9}
\begin{document}

\maketitle

%TODO mandatory: add short abstract of the document
\begin{abstract}
The main goal of this article is to present a direct approach for 
  the formula giving the long-term
  apparent hashrates of Selfish Mining strategies using only elementary
  probabilities and combinatorics, more precisely, Dyck words. We can avoid
  computing stationary probabilities on Markov chain, nor stopping times for
  Poisson processes as in previous analysis. We do apply these
  techniques to other bockwithholding strategies in Bitcoin, and then, we consider 
  also selfish mining in Ethereum.
\end{abstract}

\section{Introduction}
\label{sec:Introduction}

\paragraph*{Background}
Selfish mining (in short SM) is a non-stop blockwithholding 
attack described in {\cite{DBLP:journals/cacm/EyalS18}} which exploits a flaw in the
Bitcoin protocol in the difficulty adjustment formula {\cite{DBLP:journals/corr/abs-1805-08281}}. The
strategy is made of attack cycles. During each attack cycle, the attacker adds
blocks to a secret fork and broadcasts them to peers with an appropriate
timing. This is a deviant strategy from the Bitcoin protocol since an honest
miner never withholds validated blocks and always mines on top of the last block of the
official blockchain {\cite{nakamoto2008bitcoin}}.

A rigorous profitability analysis that incorporates time considerations 
was done in {\cite{DBLP:journals/corr/abs-1805-08281}}. The  objective function based on sound
economics principles that allows the comparison of profitabilities of different 
mining strategies with repetition is the \textit{Revenue Ratio} $\frac{\mathbb{E} [R]}{\mathbb{E} [T]}$
where $R$ and $T$ are respectively the revenue and the duration time 
per attack cycle. A blockcwithholding attack slows down the production of blocks, hurting the 
profitability per unit time of all miners, including the attacker. Only after a difficulty adjustment, 
the attack can become profitable. The  mean duration time of block production becomes
equal to $\mathbb{E} [L] \cdot \tau_B$ where $L$ is the number of blocks
added to the official blockchain by the network per attack cycle and $\tau_B =
600$ sec. is the mean validation time of a block in Bitcoin network {\cite{DBLP:journals/corr/abs-1811-09322}}. 
For Ethereum $\tau_E$ is around $12$ sec. (in what follows, we use subscript $B$ or $E$ depending on 
which network we consider).

The Revenue Ratio becomes proportional
to the long-term \textit{apparent hashrate} of the strategy $\tilde{q} =
\frac{\mathbb{E} [Z]}{\mathbb{E} [L]}$ where $Z$ is the number of blocks
added by the attacker to the official blockchain per attack cycle. 
This apparent hashrate becomes a proxy for the Revenue Ratio and can be used as 
a benchmark for profitability, but only after a difficulty adjustment. Several
methods have been devised to compute \ $\tilde{q}$. The original approach from {\cite{DBLP:journals/cacm/EyalS18}} 
uses a Markov model. Then the stationary probability is computed and used to compute the long term 
apparent hashrate. In {\cite{DBLP:journals/corr/abs-1805-08281}} we
use Martingale techniques and consider Poisson processes and associated
stopping times in order to compute the Revenue Ratio, and also the expected number of blocks $\mathbb{E} [Z]$ 
added by the attacker to the blockchain per attack cycle. The Revenue Ratio is computed 
at once using Doob's Stopping Time Theorem for Martingales. This last method has the advantage to 
compute the correct profitability analysis directly, not by means of the proxy of the long term apparent 
hashrate. For example, we can compute how long it takes for the attacker to have profit, something 
that is impossible to compute with the old Markov chain model. Moreover, with the Martingale techniques we clearly identify 
the difficulty adjustment formula as the origin of the vulnerability of the
protocol. A Bitcoin Improvement Proposal (BIP) was proposed in {\cite{DBLP:journals/corr/abs-1805-08281}} to prevent blockwithholding attacks. It consists 
in incorporating orphan blocks in the computation of the apparent hashrate of the network, and this is done by signaling orphan blocks.
Something similar is done in Ethereum where rewards are given to some orphan blocks (``uncle'' blocks). The goal was to favor mining 
decentralization.

\paragraph*{Main goal}

In this article we present a direct combinatorial approach for the direct computation of the apparent 
hashrate for different blockwithholding strategies in Bitcoin and Ethereum. These formulas are sometimes complicated,
so it is remarkable that such a direct approach is possible. We don't need to use Markov chain, nor 
Martingale theory, and only elementary combinatorics using Dyck words. This analysis does not provide the full strength 
of the Martingale theory approach, but provides the basic formulas to estimate the long term apparent hashrates, and 
hence the profitabilities of the different strategies. The situation in Ethereum is combinatorially more complex 
due to the reward of ``uncle'' blocks and their signaling, which gives a larger spectrum of possible strategies. 
Our combinatorial approach also gives closed-form formulas for the apparent hashrate of one of the most effective strategy.

\paragraph*{Notation}

As usual, the relative hashrate of the honest miners (resp. attacker) is denoted by $p$
(resp. $q$) and $\gamma$ is its \textit{connectivity} to the network. We have $p + q = 1$, $q < \frac{1}{2}$ and $0 \leq \gamma
\leq 1$. When a competition occurs between two blocks or
two forks, $\gamma$ is the fraction of the honest miners who mine on top
of a block validated by the attacker.

\medskip

We will make use of Catalan numbers and Dyck words. 
Catalan numbers can be defined by
$$
C_n=\frac{1}{2n+1}\binom{2n}{n}=\frac{(2 n) !}{n! (n + 1) !}
$$
and their generating series is 
$$
C (x) = \sum_{n=0}^{+\infty} C_n x^n=\frac{1 - \sqrt{1 - 4 x}}{2 x}
$$
A Dyck word is a string (word) composed by two letters $X$ and $Y$ such that no initial segment of the string contains 
more $Y$'s than $X$'s. The relation with Catalan numbers is that the $n$-th Catalan number is the number of Dyck words 
of length $2n$. We refer to {\cite{MR2526440}} for more  properties and background material. 

\section{Selfish mining}
\label{sec:Selfish}

An attack cycle for the SM strategy (see \cite{DBLP:journals/corr/abs-1805-08281}) can be described as a sequence $X_0 \ldots
X_n$ with $X_i \in \{S, H\}$. The index $i$ indicates the $i$-th block
validated since the beginning of the cycle and the letter $S$, resp. $H$, indicates that the selfish, resp. honest, 
miner has discovered this block. From this labelling we will get the relation with Dyck words.

\begin{example}
  The sequence $\text{SSSHSHH}$ means that the selfish miner has been first to
  validate three blocks in a row, then the honest miners have mined one, then
  the selfish miner has validated a new one and finally the honest miners have
  mined two blocks. At this point, the advantage of the selfish miner is only
  of one block. So according to the SM strategy, he decides to publish his
  whole fork and ends his attack cycle. In that case, we have $L = Z = 4$.
\end{example}

We are interested in the distribution of the random variable $L$.

\begin{theorem}
  We have $\mathbb{P} [L = 1] = p, \mathbb{P} [L = 2] = pq + pq^2$ and for
  $n \geqslant 3$, $\mathbb{P} [L = n] = pq^2  (pq)^{n - 2} C_{n - 2}$ where
  $C_n$ is the $n$-th Catalan number.
\end{theorem}

\begin{proof}
  For $n \geqslant 3$, we note that $\{L = n\}$ is a collection of sequences
  of the form $w = \text{SSX}_1 \cdots X_{2 (n - 2)} H$ with $X_i \in \{ S, H
  \}$ for all $i$, such that if $S$ and $H$ are respectively replaced by the
  brackets ``(`` and ``)'' then, $X_1 \cdots X_{2 (n - 2)}$ is a Dyck word
  (i.e. balanced parenthesis) with length $2 (n - 2)$ (see {\cite{MR2526440}}). The
  number of letters ``$S$'' (resp. ``$H$'') in $w$ is $n$ (resp. $n - 1$). So,
  we get $\mathbb{P} [L = n] = p^{n - 1} q^n C_{n - 2}$ (see {\cite{MR2526440}}).
  Finally, from the  observation that $\{L = 1\} = \{H\}, \{L = 2\} = \{\text{SSH},
  \text{SHS}, \text{SHH}\}$, the result follows.
\end{proof}

\begin{corollary}
  \label{el}We have $\mathbb{E} [L] = 1 + \frac{p^2 q}{p - q}$
\end{corollary}

\begin{proof}
  This formula results from the well know relations from {\cite{DBLP:journals/corr/abs-1808-01041}}
  \begin{align}
    \Sigma_{n \geqslant 0} p (pq)^n C_n & = 1  \label{pc}\\
    \Sigma_{n \geqslant 0} np (pq)^n C_n & = \frac{q}{p - q}  \label{ec}
  \end{align}
\end{proof}

We can now compute the apparent hashrate.

\begin{theorem}
  \label{hashratesm}The long-term apparent hashrate of the selfish miner in
  Bitcoin is
  \[ \tilde{q}_B = \frac{[(p - q) (1 + pq) + pq] q - (p - q) p^2 q (1 -
     \gamma)}{pq^2 + p - q} 
  \]
\end{theorem}

\begin{proof}
  When $L \geqslant 3$ we are in the situation where all blocks validated by
  the selfish miner end-up in the official blockchain. So, we have $Z = L$. If $L =
  1$, then we have $Z = 0$. Moreover, we have $Z (\text{SSH}) = Z (\text{SHS}) = 2$ and $Z
  (\text{SHH}) = 0$ (resp. $1$) with probability $1 - \gamma$ (resp. $\gamma$).
  So, we have
  \begin{align*}
    \mathbb{E} [Z] & =  \mathbb{E} [L] - p - p^2 q \gamma - 2 p^2 q (1 -
    \gamma)\\
    & =  \mathbb{E} [L] - (p + p^2 q + p^2 q (1 - \gamma))
  \end{align*}
  Using Corollary \ref{el} we get,
  \begin{align*}
    \frac{\mathbb{E} [Z]}{\mathbb{E} [L]} & = \frac{p^2 q + p - q - (p -
    q)  (p + p^2 q + p^2 q (1 - \gamma))}{pq^2 + p - q}\\
    & =  \frac{[(p - q) (1 + pq) + pq] q - (p - q) p^2 q (1 - \gamma)}{pq^2
    + p - q}
  \end{align*}
  This is Proposition 4.9 from {\cite{DBLP:journals/corr/abs-1805-08281}} which is 
  another form of Formula (8) from {\cite{DBLP:journals/cacm/EyalS18}}.
\end{proof}

\section{Stubborn Mining}

We consider now two other block withholding strategies described in
{\cite{DBLP:conf/eurosp/NayakKMS16}}. .

\subsection{Equal Fork Stubborn Mining}

In this strategy, the attacker never tries to override the official
blockchain but, when possible, he broadcasts the part of his secret fork
sharing the same height as the official blockchain as soon as the honest
miners publish a new block. The attack cycle ends when the attacker has been
caught-up and overtaken by the honest miners by one block
{\cite{DBLP:journals/corr/abs-1808-01041,DBLP:conf/eurosp/NayakKMS16}}. We show that the distribution of $L - 1$ is what we
defined as a $(p, q)$-Catalan distribution of first type in {\cite{DBLP:journals/corr/abs-1808-01041}}.

\begin{theorem}
  \label{thepefsm}For $n \geqslant 0$ we have $\mathbb{P} [L = n + 1]= p (pq)^n C_n$.
\end{theorem}

\begin{proof}
  For $n \geqslant 0$, $\{L = n + 1\}$ corresponds to
  sequences of the form $w = X_1 \cdots X_{2 n} H$ with $X_i \in \{ S, H \}$
  for all $i$, such that if $S$ and $H$ are respectively replaced by the
  brackets ``(`` and ``)'' then, $X_1 \cdots X_{2 n}$ is a Dyck word with
  length $2 n$.
\end{proof}

\begin{corollary}
  \label{elefsm}We have $\mathbb{E} [L] = \frac{p}{p - q}$
\end{corollary}

\begin{proof}
  Follows from (\ref{pc}) and (\ref{ec}).
\end{proof}

\begin{theorem}
  \label{thez}The long-term apparent hashrate of a miner following the
  Equal-Fork Stubborn Mining strategy is given by 
  $$
  \tilde{q} = \frac{q}{p} -
  \frac{(1 - \gamma)  (p - q)}{\gamma p}  (1 - pC ((1 - \gamma) pq))
  $$
\end{theorem}

\begin{proof}
  In an attack cycle, all the honest blocks except the last one have a
  probability $\gamma$ to be replaced by the attacker. So, we have $\mathbb{E} [Z|L =
  n + 1] = n + 1 - \frac{1 - (1 - \gamma)^{n + 1}}{\gamma}$ (see Lemma B.1 in
  {\cite{DBLP:journals/corr/abs-1808-01041}}). Conditioning by $\{ L = n + 1 \}$ for $n \in \mathbb{N}$
  and using Theorem \ref{thepefsm}, we get
  \[ \mathbb{E} [Z] = \frac{q}{p - q} - \frac{1 - \gamma}{\gamma}  (1 - pC
     ((1 - \gamma) pq)) \]
 and the result follows.
\end{proof}

\subsection{Lead Stubborn Mining}

This strategy is similar to the selfish mining strategy but this time the attacker
takes the risk of being caught-up by the honest miners. When this happens,
there is a final competition between two forks sharing the same height. when
the competition is resolved, a new attack cycles starts. In this case, the
distribution of $L - 1$ turns out to be a $(p, q)$-Catalan distribution
of second type as defined in {\cite{DBLP:journals/corr/abs-1808-01041}}.

\begin{theorem}
  \label{theplsm}We have $\mathbb{P} [L = 1] = p$ and for $n \geqslant 1$, $
  \mathbb{P} [L = n + 1] = (pq)^n C_{n - 1}$.
\end{theorem}

\begin{proof}
  We have $\{ L = 1 \} = \{ H \}$ and for $n \geqslant 0$, the condition $\{L =
  n + 1\}$ corresponds to sequences of the form $w = \text{SX}_1 \cdots
  X_{2 (n - 1)} \text{HY}$ with $X_1, \ldots, X_{2 (n - 1)}, Y \in \{ S, H \}$
  and such that if $S$ and $H$ are respectively replaced by the brackets ``(``
  and ``)'' then, $X_1 \cdots X_{2 (n - 1)}$ is a Dyck word with length $2 (n
  - 1)$.
\end{proof}

\begin{corollary}
  \label{ellsm}We have $\mathbb{E} [L] = \frac{p - q + pq}{p - q}$
\end{corollary}

\begin{proof}
  Follows from (\ref{pc}) and (\ref{ec}).
\end{proof}

By repeating the same argument as in the proof of Theorem \ref{thez} for the
computation of $\mathbb{E} [Z]$, we obtain the following theorem
{\cite{DBLP:journals/corr/abs-1808-01041}}.

\begin{theorem}
  The long-term apparent hashrate of a miner following the Lead Stubborn
  Mining strategy is given by 
  $$
  \tilde{q} = \frac{q (p + pq - q^2)}{p + pq - q}
  - \frac{pq (p - q)  (1 - \gamma)}{\gamma} \cdot \frac{1 - p (1 - \gamma) C
  ((1 - \gamma) pq)}{p + pq - q}
  $$
\end{theorem}

We plot regions in the parameter space $(q, \gamma) \in [0, 0.5] \times [0, 1]$ according to
which strategy is more profitable. We get Figure \ref{Figure1} {\cite{DBLP:journals/corr/abs-1808-01041}} (HM honest mining, 
SM selfish mining, LSM Lead Stubborn mining, EFSM Equal Fork Stubborn mining).

\begin{figure}
%\centering
\includegraphics[width=0.55\linewidth]{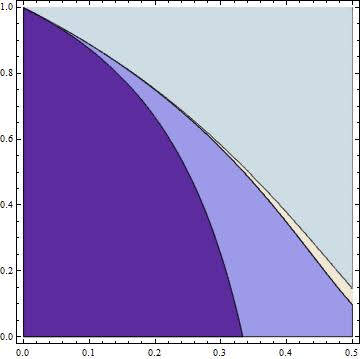}
\caption{From left to right: HM, SM, LSM, EFSM.}
\label{Figure1}
\end{figure}

\section{Selfish mining in Ethereum}

Ethereum is a cryptocurrency based on a variation of the GHOST protocol
{\cite{DBLP:conf/fc/SompolinskyZ15}}. The reward system is different than in Bitcoin, and 
this introduces a supplementary complexity in the analysis of block withholding strategies. 
Contrary to Bitcoin, mined orphan
blocks can be rewarded like regular blocks, with a reward smaller than regular
blocks. The condition for an orphan block to get a reward is to be an ``uncle''
referred by a ``nephew'' which is ``not too far''. By definition, an ``uncle''
is a stale block whose parent belongs to the main chain and a ``nephew'' is a
regular block which refers to this ``uncle''. ``Not too far'' means that the
distance $d$ between the uncle and the nephew is less than some parameter value $n_1$. 
The distance is the number of blocks which separates the nephew to
the uncle's parent in the main chain. When this situation occurs, the nephew
gets an additional reward of $\pi b$ and the uncle gets a reward $K_u (d)
b$ where $b$ denotes the coinbase in Ethereum. Today's parameter values are $n_1 = 6$, $K_u
(d) = \frac{8 - d}{8} \cdot \mathbf{1}_{1 \leqslant d \leqslant 6}$,
$\pi = \frac{1}{32}$ and $b = 2$ ETH {\cite{DBLP:journals/corr/abs-1901-04620}}.

\medskip

There is little published research on selfish mining in Ethereum except for \cite{DBLP:journals/corr/abs-1901-04620} and  \cite{DBLP:conf/eurosp/RitzZ18} 
based on numerical simulations. In \cite{DBLP:journals/corr/abs-1901-04620}, through a Markov chain approach, a non-closed infinite 
double sum is given for the apparent hash rate of the attacker.

\medskip
The general study of selfish mining in Ethereum is complex because equivalent selfish mining strategies in Bitcoin are no longer 
equivalent for Ethereum. The attacker can choose to refer or not uncle  blocks. Referring uncle blocks provides an extra revenue but hurts 
the main goal of selfish mining of lowering the difficulty. Also, he can choose to create artificially more uncles by broadcasting the part of 
his secret fork sharing the same height as the public blockchain of the honest miners. All this different strategies are analyzed in \cite{SME}.
In the present article we restrict to a couple of strategies.

\medskip
In the strategy studied in \cite{DBLP:journals/corr/abs-1901-04620}, the attacker creates as many 
uncles as possible and tries to refer all of them. In \cite{SME}, we prove that this strategy is not optimal and is less  
profitable than the strategy we study in this article, for which we obtain a closed 
form formula for the apparent hashrate of the attacker using only elementary combinatorics. 

In the strategy we consider, the attacker never broadcasts his fork, which remains secret until he is on the edge of being 
caught-up by the honest miners or is actually caught up (this last case can only occur when the attack cycle starts 
with SH). In addition, the attacker always refers to all possible uncle blocks.

\bigskip

We denote by $R$ the revenue by cycle of the selfish miner following this strategy. 
We have $R = R_s + R_u + R_n$ where $R_s$ is
the revenue coming from ``static'' blocks in the main chain i.e., $R_s = Zb$,
$R_u$ is the revenue coming from uncles and $R_n$ is the additional revenue
coming from nephews.

\begin{remark}
  We always have $R_u = 0$ except when the attack cycle is $\text{SHH}$ and
  the last block mined by the honest miners has been mined on top of an honest
  block. In that case, the first block mined by the selfish miner is referred
  by the second block of the honest miners.
\end{remark}
It follows from this remark that
\begin{equation}
  \mathbb{E} \left[ \frac{R_u}{b} \right] = p^2 q (1 - \gamma) K_u (1)
  \label{eru}
\end{equation}

It remains to compute $\mathbb{E}[R_n]$.

\begin{definition}\label{gedef}
  If $\omega$ is an attack cycle, we denote by $U (\omega)$  (resp. $U_s
  (\omega)$, $U_h (\omega)$) the random variable counting the number of 
  uncles created during the cycle $\omega$ which
  are referred by nephew blocks (resp. nephew blocks mined by the selfish
  miner, nephew blocks mined by the honest miners) in the cycle $\omega$ or in a later
  attack cycle.
  
  We denote by $V (\omega)$ the random variable counting the number of
  uncles created during the cycle $\omega$ and are referred by nephew blocks
  (honest or not) in an attack cycle strictly after $\omega$.
\end{definition}

\begin{proposition}\label{proeu2}
  We have $\mathbb{E} [U] = q - q^{n_1 + 1}$.
\end{proposition}

\begin{proof}
  We have $U = 0$ if and only if the attack cycle is H or if it starts with $n_1 + 1$ blocks of type S.
  Otherwise, we have $U = 1$. So, $\mathbb{E} [U]  =  \mathbb{P} [U > 0] =  1 - (p + q^{n_1 + 1}) =  q - q^{n_1 + 1}$
%  \begin{equation*}
%    \mathbb{E} [U]  =  \mathbb{P} [U > 0] =  1 - (p + q^{n_1 + 1}) =  q - q^{n_1 + 1}
%  \end{equation*}
\end{proof}

We compute now $\mathbb{E} [V]$

\begin{proposition}\label{peus}
  We have $\mathbb{E}[V] = p q^{2}\cdot \frac{1-(p q)^{n_1 - 1}}{1-p q}$.
\end{proposition}

\begin{proof}
We have $V=1$ if and only if the attack cycle $\omega$ is SS..SH..H with $2\leq k\leq n_1$ S. In that case, 
the first block H is an uncle that will be referred by the first future official block in the attack cycle after $\omega$. Otherwise, $V=0$. 
So, $\mathbb{E}[V] = pq^{2}+\ldots+p^{n_1-1}q^{n_1}$, and we get the result.
\end{proof}

\begin{proposition}\label{euh2}
  We have $\mathbb{E}[U_h] = p^{2}q+\left( p + (1-\gamma)p^{2}q\right) p q^{2}\cdot \frac{1-(p q)^{n_1 - 1}}{1-p q}$.
\end{proposition}

\begin{proof}
  Let $\omega$ be an attack cycle and let $\omega'$ be the attack cycle after $\omega$.
  If $U_h^{(1)} (\omega)$ (resp. $U_h^{(2)} (\omega)$) 
  counts the number of uncles referred by honest nephews only present in $\omega$ (resp. in $\omega'$),
  then we have $U_h = U_h^{(1)} + U_h^{(2)}$. Moreover, $U_h^{(1)} (\omega)={\bf 1}_{\omega=\text{SHH}}$ and  
  $U_h^{(2)} (\omega)={\bf 1}_{\omega'\in E}\cdot V(\omega)$ where $E$ is the event that $\omega'$ is either H or SHH with a second
  honest block mined on top of the first honest block. Hence we get the result by taking expectations since $\omega$ and $\omega'$ are
  independent.
\end{proof}

\begin{corollary}\label{coeub}
  We have
\begin{equation}\label{ern}
  \mathbb{E} \left[ \frac{R_n}{\pi} \right] = q^2 (1+p) - q^{n_1 + 1} 
  - \left( p + (1-\gamma)p^{2}q\right) p q^{2}\cdot \frac{1-(p q)^{n_1 - 1}}{1-p q}
  \end{equation}
\end{corollary}

\begin{proof}
We have $\mathbb{E} [U_s] = \mathbb{E} [U] - \mathbb{E} [U_h]$ and we use Proposition \ref{proeu2} and Proposition \ref{euh2}.
\end{proof}

We can now compute the apparent hashrate of the selfish miner in Ethereum. We have two cases to consider: 
The old difficulty adjustment formula (similar to the one  in Bitcoin), and the current difficulty adjustment formula that 
takesinto account referred uncles.
\begin{theorem}
  The long term apparent hashrate $\tilde{q}_{E, 0}$ of the selfish miner in Ethereum 
  with its old difficulty adjustment formula is given by 
  $\tilde{q}_{E, 0} = \tilde{q}_B + \tilde{q}_u K_u (1) + \tilde{q}_n \pi$ with

  \begin{align*}
    \tilde{q}_u & =  \frac{p^2 q (1 - \gamma) (p - q)}{p - q + p^2
    q}\\
    \tilde{q}_n & =  \frac{(p - q) \left(
   q^2 (1+p) - q^{n_1 + 1} 
  - \left( p + (1-\gamma)p^{2}q\right) p q^{2}\cdot \frac{1-(p q)^{n_1 - 1}}{1-p q}\right)}{p - q + p^2 q}
  \end{align*}
The long term apparent hashrate $\tilde{q}_{E}$ of the selfish miner in Ethereum with its current difficulty adjustment formula is
\begin{equation*}
\tilde{q}_{E}=\tilde{q}_{E, 0}\cdot \xi
\end{equation*} 
where 
$$
\xi = \frac{p - q + p^2 q}{p^2 q + (p-q)\bigl(1 + q - q^{n_1+1} \bigr)}
$$
\end{theorem}

\begin{proof}
We have $\tilde{q}_{E, 0}=\frac{\mathbb{E} [R]}{\mathbb{E} [L]}$ 
and $\tilde{q}_{E}=\frac{\mathbb{E} [R]}{\mathbb{E} [L] +\mathbb{E} [U]}$
We then use Proposition \ref{proeu2}, (\ref{eru}), (\ref{ern}) and the formula for $\tilde{q}_{B}$ in Theorem \ref{hashratesm}.
\end{proof}

We can now compare this strategy to selfish mining in Bitcoin. Observe that $\tilde{q}_{E, 0}>\tilde{q}_{B}$, where $\tilde{q}_{B}$
is the long term apparent hashrate of the Bitcoin selfish miner. Therefore, the minimal threshold $q_{\min}$ such that the inequality 
$\tilde{q} > q$ for $q > q_{\min}$ is always lower in Ethereum with its old adjustment formula than in Bitcoin. 
This is due to the particular reward system that indeed favors selfish mining as we have proved. 
Notice also that when $q > q_{\min}$, the attack is profitable 
faster in Ethereum than in Bitcoin because of another difference in the protocols: In Ethereum the difficulty is updated at each block and 
in Bitcoin only after $2016$ blocks.

\medskip

Figure \ref{Figure2} plots the regions in parameter space $(q, \gamma) \in [0, 0.5] \times [0, 1]$ where
each strategy HM or SM is more profitable. We find $q_{\min} \approx 9.5\%$ 
when $\gamma = 0$.

\begin{figure}
%\centering
\includegraphics[width=0.55\linewidth]{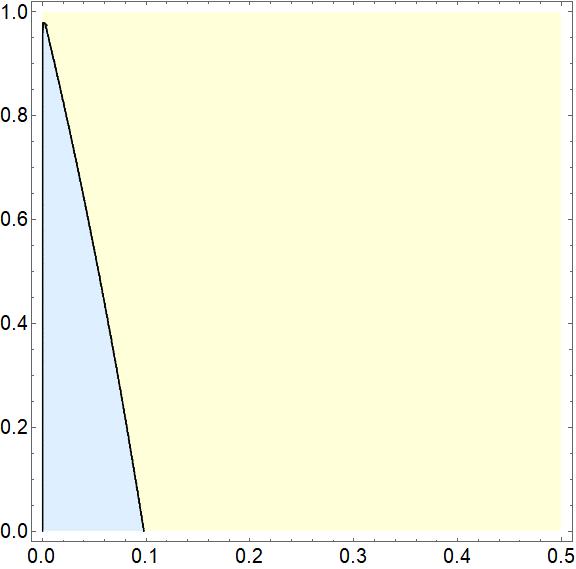}
\caption{HM vs. SM Ethereum old difficulty adjustment.}
\label{Figure2}
\end{figure}

Now, Ethereum with its new difficulty adjustment formula is more resilient to selfish mining.
Figure \ref{Figure3} plots the region in parameter space $(q, \gamma) \in [0, 0.5] \times [0, 1]$ where each
strategy HM or SM is more profitable.  
%In Figure \ref{Figure4} we have compared both threshold curves for Bitcoin and Ethereum.

\begin{figure}
%\centering
\includegraphics[width=0.55\linewidth]{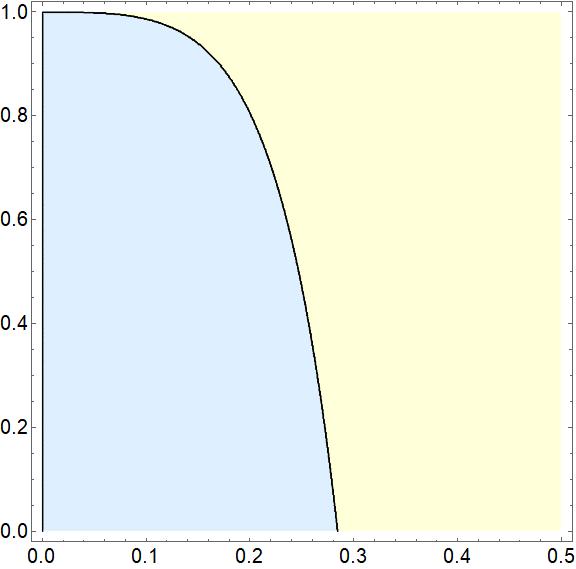}
\caption{HM vs. SM Ethereum new difficulty adjustment.}
\label{Figure3}
\end{figure}

We note that Bitcoin is more resilient to selfish mining when the relative hashrate of the attacker is high, but we have 
the opposite for smaller relative hashrates.
% Ethereum is more resilient to selfish mining when the relative hashrate of the attacker is small. On the other hand,
% Bitcoin is more resistant to selfish mining when the relative hashrate of the attacker is high. 
%This does not mean that selfish mining in Ethereum is more profitable than in Bitcoin. 
% It depends on the cost of mining and on the exchange rate ETH/BTC.
This means that when the relative hashrate of the attacker is small (resp. high) then, the 
connectivity of the attacker should be higher (resp. lower) in Ethereum than in Bitcoin for the attack to be profitable. 
Figure \ref{Figure4} compares the thresholds curves between HM and SM in Bitcoin and Ethereum.

\begin{figure}
%\centering
\includegraphics[width=0.55\linewidth]{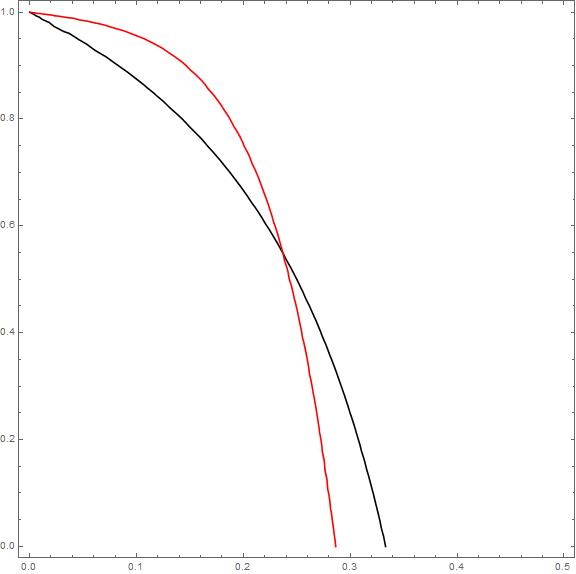}
\caption{HM vs. SM in Bitcoin and Ethereum.}
\label{Figure4}
\end{figure}

\section{Conclusions.}

We have computed closed-form formulas for the long term apparent hashrate of different
blockwithholding strategies for Bitcoin and Ethereum using only elementary combinatorics, 
Dyck words, Catalan numbers, and their properties. Although this approach does not 
provide a complete analysis of the profitability of the strategies, as for example the time
it takes to the strategy to become profitable, this minimalist approach is 
sufficient to compare profitabilities in the long run. In the strategies studied we have 
show the impact of the different reward system.  For these strategies, depending on given parameters $(q,\gamma)$, relative 
hashrate and connectivity of the attacker, we have determined which network is more resilient to 
selfish mining attacks.

\bibliography{SMDyck}

\end{document}